\makeatletter
  \providecommand*\input@path{}
  \newcommand\addinputpath[1]{
  \expandafter\def\expandafter\input@path
  \expandafter{\input@path{#1}}}
  \addinputpath{importfile/}
  \newcommand{\Rmnum}[1]{\expandafter\@slowromancap\romannumeral #1@}
  \newcommand{\thickhline}{%
    \noalign {\ifnum 0=`}\fi \hrule height .6pt
    \futurelet \reserved@a \@xhline
  }

  \def\email#1{\emailname: #1}
  \def\emailname{E-mail}%

  \def\acknowledgement{\par\addvspace{17pt}\small\rmfamily
  \trivlist\if!\ackname!\item[]\else
  \item[\hskip\labelsep
  {\bfseries\ackname}]\fi}
  
  \newenvironment{acknowledgements}{\begin{acknowledgement}}
  {\end{acknowledgement}}
  \def\trans@english{\switcht@albion}
  \def\switcht@albion{\def\ackname{Acknowledgements}%
  }\switcht@albion
\makeatother

\documentclass[12pt]{article}
\usepackage[left=2.50cm,right=2.50cm,top=2.50cm,bottom=2.50cm]{geometry} 
\usepackage{mathrsfs}
\usepackage{float}
\usepackage{cite}
\usepackage[misc]{ifsym}
\usepackage{makecell}
\usepackage{slashbox}
\usepackage{amsmath} 
\usepackage{amsfonts} 
\usepackage{amssymb} 
\usepackage{amsthm}
\usepackage[english]{babel}
\usepackage{graphicx}   
\usepackage{pict2e}
\usepackage{url}       
\usepackage{bm}         
\usepackage[justification=centering]{caption}
\usepackage{multirow}
\usepackage{booktabs}
\usepackage{algorithm}
\usepackage{algorithmic}
\usepackage[colorlinks,linkcolor=blue,anchorcolor=blue,citecolor=blue]{hyperref}

\hypersetup{colorlinks,bookmarks,unicode}
\hypersetup{urlcolor=blue}

\def\keywordname{{\bfseries Keywords}}%
\def\keywords#1{\par\addvspace\medskipamount{\rightskip=0pt plus1cm
\def\and{\ifhmode\unskip\nobreak\fi\ $\cdot$
}\keywordname\enspace\ignorespaces#1\par}}

\newcommand{\tl}{\textnormal}
\newcommand{\rank}{\textrm{Rank}}
  
\newcommand{\supp}{\textnormal{Supp}}

\newcommand{\gaubinom}[2]{\begin{bmatrix}#1\\#2\end{bmatrix}_q}

\newtheorem{theorem}{Theorem}

\newtheorem{proposition}{Proposition}

\theoremstyle{remark}
\newtheorem{definition}[theorem]{Definition}
\newtheorem{remark}{Remark}

\date{}

\begin{document}

\title{Two Public Key Cryptosystems Based on Expanded Gabidulin Codes %\thanks{Grants or other notes
}

\author{Wenshuo Guo\thanks{Wenshuo Guo is with the Chern Institute of Mathematics and LPMC, Nankai University, Tianjin 300071, China. \email{ws\_guo@mail.nankai.edu.cn}} 
\ and 
Fang-Wei Fu\thanks{Fang-Wei Fu is with the Chern Institute of Mathematics and LPMC, Nankai University, Tianjin 300071, China. \email{fwfu@nankai.edu.cn}}}

\maketitle

\begin{abstract}
This paper presents two public key cryptosystems based on the so-called expanded Gabidulin codes, which are constructed by expanding Gabidulin codes over the base field. Exploiting the fast decoder of Gabidulin codes, we propose an efficient algorithm to decode these new codes when the noise vector satisfies a certain condition. Additionally, these new codes have an excellent error-correcting capability because of the optimality of their parent Gabidulin codes. With different masking techniques, we give two encryption schemes by using expanded Gabidulin codes in the McEliece setting. Being constructed over the base field, these two proposals can prevent the existing structural attacks using the Frobenius map. Based on the distinguisher for Gabidulin codes, we propose a distinguisher for expanded Gabidulin codes by introducing the concept of the so-called twisted Frobenius power. It turns out that the public code in our proposals seems indistinguishable from random codes under this distinguisher. Furthermore, our proposals have an obvious advantage in public key representation without using the cyclic or quasi-cyclic structure compared to some other code-based cryptosystems. To achieve the security of 256 bits, for instance, a public key size of 37583 bytes is enough for our first proposal, while around 1044992 bytes are needed for Classic McEliece selected as a candidate of the third round of the NIST PQC project.
\end{abstract}

\section{Introduction}
Over the past decades, cryptosystems based on coding theory have been drawing more and more attention due to the rapid development of quantum computers. The first code-based cryptosystem, known as McEliece cryptosystem \cite{McEliece1978public} based on Goppa codes, was proposed by McEliece in 1978. The principle for McEliece's proposal is to first encode the plaintext with a random generator matrix of the distorted Goppa code and then add some random errors. Since then various studies\cite{lee1988observation,canteaut2000cryptanalysis,loidreau2001weak,faugere2013distinguisher,couvreur2016polynomial} have been made to investigate the security of McEliece cryptosystem.

In addition to potential resistance against quantum computer attacks, McEliece cryptosystem has pretty fast encryption and decryption procedures. However, this system has never been used in practice due to the large public key size. To overcome this problem, some variants were proposed one after another. For instance, the authors in \cite{loidreau2000strengthening} proposed to use the automorphism groups of Goppa codes to build decodable error patterns of larger weight, which greatly enhances the system against decoding attacks. By doing this, smaller codes are allowed in the design of encryption schemes to reduce the public key size. Unfortunately, this variant was shown to be vulnerable against the chosen-plaintext attacks proposed in \cite{kobara2003one}. In \cite{misoczki2009compact}, the authors proposed the family of quasi-dyadic Goppa codes, which admit a very compact representation of parity-check or generator matrix, for efficiently designing syndrome-based cryptosystems. However, the authors in \cite{faugere2010algebraic} mounted an efficient key-recovery attack against this variant for almost all the proposed parameters.

Besides endowing Goppa codes with some special structures, replacing Goppa codes with other families of codes is another approach to shorten the public keys. For instance, Niederreiter \cite{Niederreiter1986Knapsack} introduced a knapsack-type cryptosystem based on generalized Reed-Solomon (GRS) codes. In Niederreiter's proposal, the message sender first converts the plaintext into a vector of fixed weight and then multiplies it with a parity-check matrix of the public code. The advantage of GRS codes consists in their optimal error-correcting capability, which enables us to reduce the public key size by exploiting codes with smaller parameters. However, this variant was proved to be insecure by Sidelnikov and Shestakov in \cite{Sidelnikov1992On} for the reason that GRS codes are highly structured.  But if we use Goppa codes in the Niederreiter setting, it was proved to be equivalent to McEliece cryptosystem in terms of security \cite{xinmei1994on}. To strengthen resistance against structural attacks, the authors in \cite{Baldi2016Enhanced} performed a column-mixing transformation instead of a simple permutation to the underlying GRS code. According to their analysis, this variant could prevent some well-known attacks, such as Sidelnikov-Shestakov attack \cite{Sidelnikov1992On} and Wieschebrink's attack \cite{Wieschebrink2010Crypt}. However, in\cite{Couvreur2014Distinguisher} the authors presented a polynomial key-recovery attack in some cases. Although one can adjust the parameters to prevent such an attack, it would introduce some other problems, such as the decryption complexity increasing dramatically and a higher request of error-correcting capability for the underlying code. In \cite{khathuria2021encryption}, the authors introduced the concept of expanded GRS codes and designed an encryption scheme by using these codes in the Niederreiter setting. However, this scheme was already partially broken by Couvreur and Lequesne in \cite{couvreur2022onthe} for the case of $\lambda=2$ and $m=3$.

In \cite{gabidulin1985theory} Gabidulin introduced a new family of rank metric codes, known as Gabidulin codes, which can be seen as an analogue of GRS codes in the rank metric. The particular appeal of rank metric codes is that the general decoding problem is much more difficult than that of Hamming metric codes \cite{chabaud1996crypt,ourivski2002new}. This inspires us to obtain a much smaller public key size by building cryptosystems in the rank metric. In \cite{gabidulin1991ideals} the authors proposed the GPT cryptosystem by using Gabidulin codes in the McEliece setting, which requires a public key size of only a few thousand bits for the security of $100$ bits. Similar to the cryptosystems based on GRS codes, the GPT cryptosystem and some of its variants \cite{gabidulin2008attacks,gabidulin2009improving,loidreau2010designing,rashwan2010smart} have been subjected to many structural attacks \cite{Gibson1996The,2008Structural,horlemann2018extension,Otmani2018Improved}. Faure and Loidreau proposed another cryptosystem \cite{faure2005new} that is quite different from the GPT proposal. The security of this scheme closely relates to the intractability of reconstructing linearized polynomials. Until the work in \cite{gaborit2018polynomial}, the Faure-Loidreau scheme had never been severely attacked. In \cite{loidreau2017new} Loidreau designed another rank metric based cryptosystem in the McEliece setting, where a column-mixing transformation was imposed to the underlying code with the inverse of an invertible matrix whose entries are taken from an $\mathbb{F}_q$-subspace of $\mathbb{F}_{q^m}$ of dimension $\lambda$. Loidreau claimed that their proposal could prevent all the existing structural attacks. However, this claim was proved to be invalidated by the authors in \cite{coggia2020security} for the case of $\lambda=2$ and the code rate being greater than $1/2$. Not long after this, the author in \cite{2020Extending} extended this attack to the case of $\lambda=3$.

In \cite{berger2017gabidulin}, Berger et al. used the so-called Gabidulin matrix codes to design cryptosystems, which can be seen as a rank metric counterpart of expanded GRS codes. Our work in the present paper is inspired by the variants \cite{khathuria2021encryption,berger2017gabidulin} and uses the so-called expanded Gabidulin codes as the underlying code. Benefitting from the optimality of their parent Gabidulin codes, these new codes have excellent capability of correcting Hamming errors. This enables us to reduce the public key size by exploiting smaller codes. Because of our proposals being constructed over the base field, all the existing structural attacks based on the Frobenius map do not work any longer.

The rest of this paper is arranged as follows. In Section \ref{section2}, notations and some basic concepts about rank metric codes and Gabidulin codes will be given. In Section \ref{section3}, we shall introduce the so-called expanded Gabidulin codes and propose an efficient algorithm to decode these new codes. Section \ref{section4} introduces two hard problems in coding theory, as well as the best known generic attacks on them. Section \ref{section5} is devoted to a formal description of our two proposals constructed by using expanded Gabidulin codes in the McEliece setting. Section \ref{section6} gives the security analysis of our proposals, including structural attacks and generic attacks. In Section \ref{section7}, we give some suggested parameters for different security levels and make a comparison on public key size with some other code-based cryptosystems. Following this, we make a few concluding remarks in Section \ref{section8}.

%Your text comes here. Separate text sections with
\section{Preliminaries}\label{section2}
In this section we first introduce some notations used throughout this paper, and recall some basic concepts about linear codes and rank metric codes. After that, we will introduce the definition of Gabidulin codes and some related results.

\subsection{Notations and basic concepts}
Let $q$ be a prime power. Denote by $\mathbb{F}_q$ the finite field with $q$ elements, and $\mathbb{F}_{q^m}$ an extension field of $\mathbb{F}_q$ of degree $m$. For two positive integers $k$ and $n$, let $\mathcal{M}_{k,n}(\mathbb{F}_q)$ denote the space of all $k\times n$ matrices over $\mathbb{F}_q$, and $GL_n(\mathbb{F}_q)$ the general linear group of all invertible matrices in $\mathcal{M}_{n,n}(\mathbb{F}_q)$. For a matrix $M\in\mathcal{M}_{k,n}(\mathbb{F}_q)$, let $\langle M\rangle_{\mathbb{F}_q}$ be the vector space spanned by the rows of $M$ over $\mathbb{F}_q$.

An $[n,k]$ linear code $\mathcal{C}$ over $\mathbb{F}_q$ is a $k$-dimensional subspace of $\mathbb{F}_q^n$. An element of $\mathcal{C}$ is called a codeword of $\mathcal{C}$. The dual code of $\mathcal{C}$, denoted by $\mathcal{C}^\perp$, is the orthogonal space of $\mathcal{C}$ under the usual inner product over $\mathbb{F}_q^n$. A matrix $G$ is called a generator matrix of $\mathcal{C}$ if its rows form a basis of $\mathcal{C}$. A generator matrix $H$ of $\mathcal{C}^\perp$ is called a parity-check matrix of $\mathcal{C}$. For a codeword $\bm{c}\in\mathcal{C}$, the Hamming weight of $\bm{c}$, denoted by $w_H(\bm{c})$, is the number of nonzero components of $\bm{c}$. The minimum Hamming distance of $\mathcal{C}$, denoted by $d_H(\mathcal{C})$, is defined as the minimum Hamming weight of nonzero codewords in $\mathcal{C}$. The minimum Hamming distance of $\mathcal{C}$ has $n-k+1$ as an upper bound, and we call $\mathcal{C}$ Maximum Distance Separable (MDS) when $d_H(\mathcal{C})$ reaches to this bound.

\subsection{Rank metric codes}
Now we recall some basic concepts about rank metric codes.
\begin{definition}
For a vector $\bm{x}\in\mathbb{F}_{q^m}^n$, the rank support of $\bm{x}$, denoted by $\supp(\bm{x})$, is defined as the linear space spanned by the components of $\bm{x}$ over $\mathbb{F}_q$. 
\end{definition}
 
\begin{definition}
For a vector $\bm{x}\in\mathbb{F}_{q^m}^n$, the rank weight of $\bm{x}$, denoted by $w_R(\bm{x})$, is defined as the dimension of $\supp(\bm{x})$ over $\mathbb{F}_q$. 
\end{definition}

\begin{definition}
For two vectors $\bm{x},\bm{y}\in\mathbb{F}_{q^m}^n$, the rank distance between $\bm{x}$ and $\bm{y}$, denoted by $d_R(\bm{x},\bm{y})$, is defined to be the rank weight of $\bm{x}-\bm{y}$.
\end{definition}

\begin{definition}
For a linear code $\mathcal{C}\subseteq \mathbb{F}_{q^m}^n$, the minimum rank distance of $\mathcal{C}$, denoted by $d_R(\mathcal{C})$, is defined to be the minimum rank weight of nonzero codewords in $\mathcal{C}$. 
\end{definition}

A linear code endowed with the rank metric is called a rank metric code. Similar to Hamming metric codes, the minimum rank distance of a rank metric code is bounded from above by the Singleton-type bound as described in the following proposition.

\begin{proposition}[Singleton-type bound]\cite{gabidulin2003reducible}
For positive integers $k\leqslant n\leqslant m$, let $\mathcal{C}$ be an $[n,k]$ rank metric code over $\mathbb{F}_{q^m}$, then the minimum rank distance of $\mathcal{C}$ with respect to $\mathbb{F}_q$ satisfies the following inequality
\[d_R(\mathcal{C})\leqslant n-k+1.\]
\end{proposition}

\begin{remark}
A rank metric code attaining the Singleton-type bound is called a Maximum Rank Distance (MRD) code. Let $\mathcal{C}\subseteq\mathbb{F}_{q^m}^n$ be an $[n,k]$ MRD code and $\bm{c}\in\mathcal{C}$ be a nonzero codeword with $w_H(\bm{c})=d_H(\mathcal{C})$. It is clear that $n-k+1\leqslant w_R(\bm{c})\leqslant w_H(\bm{c})\leqslant n-k+1$. Then $w_H(\bm{c})=n-k+1$, which implies that an MRD code is MDS in the Hamming metric.
\end{remark}

\subsection{Gabidulin codes}
For an integer $i$, we denote by $[i]=q^i$ the $i$-th Frobenius power. Under this notation, $\alpha^{q^i}$ can be simply written as $\alpha^{[i]}$ for any $\alpha\in\mathbb{F}_{q^m}$. For a vector $\bm{v}\in\mathbb{F}_{q^m}^n$, we denote by $\bm{v}^{[i]}$ the $i$-th component-wise Frobenius power of $\bm{v}$. For a linear code $\mathcal{C}\subseteq\mathbb{F}_{q^m}^n$, the $i$-th Frobenius power of $\mathcal{C}$ is defined as $\mathcal{C}^{[i]}=\{\bm{c}^{[i]}:\bm{c}\in\mathcal{C}\}$.

\begin{definition}[Gabidulin codes]\label{definition1}
For positive integers $k\leqslant n\leqslant m$, let $\bm{g}=(g_1,\ldots,g_n)\in\mathbb{F}_{q^m}^n$ with $w_R(\bm{g})=n$. The $[n,k]$ Gabidulin code $\tl{Gab}_{n,k}(\bm{g})$ generated by $\bm{g}$ is defined to be a linear code having a generator matrix of the form
\[G=
\begin{pmatrix}
g_1&g_2&\cdots&g_n\\
g_1^{[1]}&g_2^{[1]}&\cdots&g_n^{[1]}\\
\vdots&\vdots&\ddots&\vdots\\
g_1^{[k-1]}&g_2^{[k-1]}&\cdots&g_n^{[k-1]}
\end{pmatrix}.
\]
\end{definition}

Similar to GRS codes in the Hamming metric, Gabidulin codes also admit an excellent error-correcting capability and simple algebraic structure. The following two theorems describe some properties of Gabidulin codes.
\begin{theorem}\cite{horlemann2015new}\label{mrd}
A Gabidulin code is an MRD code. In other words, the minimum rank weight of $\tl{Gab}_{n,k}(\bm{g})$ attains the Singleton-type bound.
\end{theorem}

This implies that the Gabidulin code $\tl{Gab}_{n,k}(\bm{g})$ can theoretically correct up to $\lfloor\frac{n-k}{2}\rfloor$ rank errors, which is an important reason for Gabidulin codes being widely used in the design of cryptosystems.

\begin{theorem}\cite{gaborit2018polynomial}
The dual of a Gabidulin code is also a Gabidulin code. Specifically, we have $\tl{Gab}_{n,k}(\bm{g})^\perp=\tl{Gab}_{n,n-k}(\bm{h}^{q^{-(n-k-1)}})$ for some $\bm{h}\in\tl{Gab}_{n,n-1}(\bm{g})^\perp$ with $w_R(\bm{h})=n$, where $\bm{h}^j$ denotes the $j$-th component-wise power of $\bm{h}$.
\end{theorem}

\section{Expanded Gabidulin codes}\label{section3}
In this section, we first introduce the definition of expanded Gabidulin codes, then investigate some of their algebraic properties. After that, we propose an efficient algorithm to decode these codes.
\subsection{Expanded Gabidulin codes}\label{sect3.1}
Note that $\mathbb{F}_{q^m}$ can be viewed as an $\mathbb{F}_q$-linear space of dimension $m$. Let $\mathcal{B}=(\alpha_1,\ldots,\alpha_m)$ be a basis of $\mathbb{F}_{q^m}$ over $\mathbb{F}_q$. For any $\alpha\in\mathbb{F}_{q^m}$, there exists $(a_1,\ldots,a_m)\in\mathbb{F}_q^m$ such that $\alpha=\sum_{i=1}^ma_i\alpha_i$. Based on this observation, we define an $\mathbb{F}_q$-linear isomorphism from $\mathbb{F}_{q^m}$ to $\mathbb{F}_q^m$ with respect to $\mathcal{B}$ as follows
\begin{align*}
\phi_\mathcal{B}:\mathbb{F}_{q^m}&\longmapsto \mathbb{F}_q^m,\\
\sum_{i=1}^ma_i\alpha_i&\longmapsto (a_1,\ldots,a_m).
\end{align*}
For a vector $\bm{v}=(v_1,\ldots,v_n)\in\mathbb{F}_{q^m}^n$, we define $\phi_\mathcal{B}(\bm{v})=(\phi_\mathcal{B}(v_1),\ldots,\phi_\mathcal{B}(v_n))\in\mathbb{F}_q^{mn}$. For a linear code $\mathcal{C}\subseteq\mathbb{F}_{q^m}^n$, we define $\phi_\mathcal{B}(\mathcal{C})=\{\phi_\mathcal{B}(\bm{c}):\bm{c}\in\mathcal{C}\}$. 

For convenience, we need to introduce a matrix representation of this transformation in \cite{berger2017gabidulin,couvreur2022onthe}. For any $\alpha\in\mathbb{F}_{q^m}^*$, there exists a unique $M_\alpha\in GL_m(\mathbb{F}_q)$ such that $\alpha\mathcal{B}^T=M_\alpha\mathcal{B}^T$. Apparently $M_\alpha$ depends on the choice of $\mathcal{B}$. For a fixed $\mathcal{B}$, we define an $\mathbb{F}_q$-linear homomorphism from $\mathbb{F}_{q^m}^*$ to $GL_m(\mathbb{F}_q)$ as $\Phi_\mathcal{B}(\alpha)=M_\alpha$. And specially, we define $\Phi_\mathcal{B}(0)=0_{m\times m}$. It is easy to verify that
\[
\Phi_\mathcal{B}(\alpha)=
\begin{pmatrix}
\phi_\mathcal{B}(\alpha\alpha_1)\\
\vdots\\
\phi_\mathcal{B}(\alpha\alpha_m)
\end{pmatrix}.
\]
For a vector $\bm{v}\in\mathbb{F}_{q^m}^n$, we define $\Phi_\mathcal{B}(\bm{v})=\begin{pmatrix}\Phi_\mathcal{B}(v_1),\ldots,\Phi_\mathcal{B}(v_n)\end{pmatrix}\in\mathcal{M}_{m,nm}(\mathbb{F}_q)$. For a matrix $M=(M_{ij})\in\mathcal{M}_{k,n}(\mathbb{F}_q)$, we define $\Phi_\mathcal{B}(M)=(\Phi_\mathcal{B}(M_{ij}))\in\mathcal{M}_{km,nm}(\mathbb{F}_q)$.

Now we give an effective method to perform this operation, which is based on the following theorem.

\begin{theorem}\cite{mullen2013handbook}\label{theorem3}
Let $\mathcal{B}=(\alpha_1,\ldots,\alpha_m)$ be a basis of $\mathbb{F}_{q^m}$ over $\mathbb{F}_q$, then there exists another basis $\mathcal{B}^*=(\alpha_1^*,\ldots,\alpha_m^*)$ such that for $1\leqslant i,j\leqslant m$ we have
\begin{equation}\label{dual}
\tl{Tr}(\alpha_i\alpha_j^*)=
\begin{cases}
1\quad \tl{for}\ i=j,\\
0\quad \tl{for}\ i\neq j,
\end{cases}
\end{equation}
where $\tl{Tr}(\cdot)$ denotes the trace function $\tl{Tr}(x)=\sum_{i=0}^{m-1}x^{q^i}$.
\end{theorem}

\begin{remark}\label{remark2}
A basis $\mathcal{B}^*$ satisfying the condition (\ref{dual}) is called a dual basis of $\mathcal{B}$. Furthermore, $\mathcal{B}^*$ is uniquely determined by $\mathcal{B}$. For any $\alpha\in\mathbb{F}_{q^m}$, let $(a_1,\ldots,a_m)\in\mathbb{F}_q^m$ such that $\alpha=\sum_{i=1}^ma_i\alpha_i$. Then we can obtain $a_j$ for $1\leqslant j\leqslant m$ by computing
\begin{align*}
\tl{Tr}(\alpha\alpha_j^*)&=\tl{Tr}(\sum_{i=1}^ma_i\alpha_i\alpha_j^*)=\sum_{i=1}^ma_i\tl{Tr}(\alpha_i\alpha_j^*)=a_j.
\end{align*}
Finally we have $\phi_\mathcal{B}(\alpha)=(\tl{Tr}(\alpha\alpha_1^*),\ldots,\tl{Tr}(\alpha\alpha_m^*))$.
\end{remark}

Now we formally introduce the concept of expanded Gabidulin codes.
\begin{definition}[Expanded Gabidulin codes]
Let $\mathcal{G}$ be an $[n,k]$ Gabidulin code over $\mathbb{F}_{q^m}$. For a basis $\mathcal{B}$ of $\mathbb{F}_{q^m}$ over $\mathbb{F}_q$, let $\phi_\mathcal{B}$ be the associated $\mathbb{F}_q$-linear isomorphism from $\mathbb{F}_{q^m}$ to $\mathbb{F}_q^m$. The expanded code of $\mathcal{G}$ induced by $\phi_\mathcal{B}$ is defined as $\widehat{\mathcal{G}}=\{\phi_\mathcal{B}(\bm{c}):\bm{c}\in\mathcal{G}\}$. Meanwhile, we call $\mathcal{G}$ the parent Gabidulin code of $\widehat{\mathcal{G}}$.
\end{definition}

It is easy to verify that $\widehat{\mathcal{G}}$ forms an $[nm,km]$ linear code over $\mathbb{F}_q$. The following proposition gives a method of constructing a generator (parity-check) matrix of an expanded Gabidulin code when a generator (parity-check) matrix of its parent Gabidulin code is known.
\begin{proposition} \cite{wu2011on}\label{theorem}
Let $\mathcal{G}\subseteq\mathbb{F}_{q^m}^n$ be an $[n,k]$ Gabidulin code. For a basis $\mathcal{B}=(\alpha_1,\ldots,\alpha_m)$ of $\mathbb{F}_{q^m}$ over $\mathbb{F}_q$, let $\widehat{\mathcal{G}}$ be the expanded code of $\mathcal{G}$ induced by $\phi_\mathcal{B}$. Then we have the following conclusions.
\begin{itemize}
\item[(1)]Let $G=
\begin{bmatrix}
\bm{g}_1^T,\ldots,\bm{g}_k^T
\end{bmatrix}^T
$
be a generator matrix of $\mathcal{G}$, then $\widehat{\mathcal{G}}$ has an $mk\times mn$ generator matrix of the form
\[\widehat{G}=\Phi_\mathcal{B}(G)=
\begin{bmatrix}\label{egabidulin1}
\phi_\mathcal{B}(\alpha_1\bm{g}_1)^T,\ldots,\phi_\mathcal{B}(\alpha_m\bm{g}_1)^T,\ldots,\phi_\mathcal{B}(\alpha_1\bm{g}_k)^T,\ldots,
\phi_\mathcal{B}(\alpha_m\bm{g}_k)^T
\end{bmatrix}^T.
\]
We call such $\Phi_\mathcal{B}(G)$ a normal generator matrix of $\widehat{\mathcal{G}}$.
\item[(2)]Let $H=\begin{bmatrix}\bm{h}_1^T,\ldots,\bm{h}_n^T\end{bmatrix}$ be a parity-check matrix of $\mathcal{G}$, then $\widehat{\mathcal{G}}$ has an $m(n-k)\times nm$ parity-check matrix of the form
\begin{align}\label{paritycheck}
\widehat{H}=\Phi_\mathcal{B}(H^T)^T=\begin{bmatrix}\phi_\mathcal{B}(\alpha_1\bm{h}_1)^T,\ldots,\phi_\mathcal{B}(\alpha_m\bm{h}_1)^T,\ldots,\phi_\mathcal{B}(\alpha_1\bm{h}_n)^T,\ldots,\phi_\mathcal{B}(\alpha_m\bm{h}_n)^T\end{bmatrix}.
\end{align}
\end{itemize}
\end{proposition}

Note that Gabidulin codes are optimal in both the Hamming metric and rank metric. However, expanded Gabidulin codes are far from optimal in the Hamming metric. Specifically, we have the following proposition.

\begin{proposition}\label{theorem5}
Let $\mathcal{G}$ be an $[n,k]$ Gabidulin code over $\mathbb{F}_{q^m}$. For a basis $\mathcal{B}$ of $\mathbb{F}_{q^m}$ over $\mathbb{F}_q$, denote by $\widehat{\mathcal{G}}$ the expanded code of $\mathcal{G}$ induced by $\phi_\mathcal{B}$. Then the minimum Hamming distance of $\widehat{\mathcal{G}}$ satisfies the following inequality
\[n-k+1\leqslant d_H(\widehat{\mathcal{G}})\leqslant m(n-k)+1.\]
In particular, with a proper choice of $\mathcal{B}$, the minimum Hamming distance of $\widehat{\mathcal{G}}$ can reach to $n-k+1$.
\end{proposition}

\begin{proof}
For any $\widehat{\bm{u}}\in\widehat{\mathcal{G}}$, there exists $\bm{u}=(u_1,\ldots,u_n)\in\mathcal{G}$ such that $\widehat{\bm{u}}=\phi_\mathcal{B}(\bm{u})$. Since $\mathcal{G}$ is MDS in the Hamming metric, then $w_H(\bm{u})\geqslant n-k+1$ for a nonzero $\bm{u}$. Let $I=\{1\leqslant i\leqslant n:u_i\neq 0\}$, then $|I|\geqslant n-k+1$. Apparently $w_H(\widehat{\bm{u}})=\sum_{i\in I}w_H(\phi_\mathcal{B}(u_i))\geqslant n-k+1$ because of $w_H(\phi_\mathcal{B}(\alpha))\geqslant 1$ for any $\alpha\in\mathbb{F}_{q^m}^*$, which implies that $d_H(\widehat{\mathcal{G}})\geqslant n-k+1$. On the other hand, by the Singleton bound for Hamming metric codes, it is clear that $d_H(\widehat{\mathcal{G}})\leqslant m(n-k)+1$.

Let $\bm{v}=(v_1,\ldots,v_n)\in\mathcal{G}$ with $w_H(\bm{v})=n-k+1$, and let $S=\{v_i\neq 0:1\leqslant i\leqslant n\}$. If $S\subseteq \mathcal{B}$, then $w_H(\phi_\mathcal{B}({\bm{v}}))=\sum_{v_i\neq 0}w_H(\phi_\mathcal{B}(v_i))=n-k+1$ because of $w_H(\phi_\mathcal{B}(\alpha))=1$ for any $\alpha\in\mathcal{B}$. This implies that $d_H(\widehat{\mathcal{G}})=n-k+1$.
\end{proof}

\subsection{Decoding expanded Gabidulin codes}\label{condition}
As for Gabidulin codes, several efficient decoding algorithms \cite{gabidulin1985theory,loidreau2005welch,richter2004error} already exist. Now we investigate the decoding problem of expanded Gabidulin codes. Our analysis shows that when the noise vector satisfies a certain condition, decoding an expanded Gabidulin code can be converted into decoding the parent Gabidulin code.

Let $\mathcal{G}\subseteq\mathbb{F}_{q^m}^n$ be an $[n,k]$ Gabidulin code having $H$ as a parity-check matrix. Let $\mathcal{B}$ be a basis of $\mathbb{F}_{q^m}$ over $\mathbb{F}_q$, we denote by $\widehat{\mathcal{G}}$ an expanded code of $\mathcal{G}$ induced by $\phi_\mathcal{B}$. Let $\bm{y}=\bm{c}+\bm{e}$ be the received word, where $\bm{c}\in\widehat{\mathcal{G}}$ and $\bm{e}=(\bm{e}_1,\ldots,\bm{e}_n)\in\mathbb{F}_q^{mn}$ is the noise vector with $\bm{e}_j=(e_{1j},\ldots,e_{mj})\in\mathbb{F}_q^m$. Let $E\in\mathcal{M}_{n,m}(\mathbb{F}_q)$ be a matrix whose $j$-th row is $\bm{e}_j$, called the error matrix corresponding to $\bm{e}$. If the following inequality holds 
\[\rank(E)\leqslant \lfloor\frac{n-k}{2}\rfloor,\] 
then we say $\bm{e}$ satisfies the decodable condition. In this situation, we can obtain a fast decoder $\mathscr{D}_{\widehat{\mathcal{G}}}$ for $\widehat{\mathcal{G}}$ to decode $\bm{y}$ by exploiting the syndrome decoder of $\mathcal{G}$.

Denote by $\widehat{H}$ a parity-check matrix of $\widehat{\mathcal{G}}$. It is easy to see that
\begin{align*}
\bm{y}\widehat{H}^T&=\bm{e}\widehat{H}^T=\sum_{j=1}^n\sum_{i=1}^me_{ij}\phi_\mathcal{B}(\alpha_i\bm{h}_j)=\phi_\mathcal{B}(\sum_{j=1}^n\sum_{i=1}^me_{ij}\alpha_i\bm{h}_j)=\phi_\mathcal{B}(\sum_{j=1}^ne_j^*\bm{h}_j)=\phi_\mathcal{B}(\bm{e}^*H^T),
\end{align*}
where $\bm{e}^*=(e_1^*,\ldots,e_n^*)\in\mathbb{F}_{q^m}^n$ with $e_j^*=\sum_{i=1}^me_{ij}\alpha_i$. It is clear that $\bm{e}^*=\mathcal{B}E^T$, then $w_R(\bm{e}^*)=\rank(E)\leqslant \lfloor\frac{n-k}{2}\rfloor$. Applying the decoder of $\mathcal{G}$ to $\phi_\mathcal{B}^{-1}(\bm{y}\widehat{H}^T)=\bm{e}^*H^T$ will lead to $\bm{e}^*$, then we can recover $\bm{e}$ by computing $\phi_\mathcal{B}(\bm{e}^*)$.

Apparently four steps are needed to decode expanded Gabidulin codes. Firstly, we shall compute the syndrome of the received word $\bm{y}$, which requires an operation of multiplying $\bm{y}$ and $\widehat{H}^T$ together with a complexity of $\mathcal{O}(m^2n(n-k))$ in $\mathbb{F}_q$. Secondly, we shall perform the inverse transformation of $\phi_{\mathcal{B}}$ to the syndrome obtained in the first step, requiring a complexity of $\mathcal{O}(mn)$ in $\mathbb{F}_{q^m}$. The third step shall call the fast decoder of the parent Gabidulin code to obtain an error vector $\bm{e}^*$ with $w_R(\bm{e}^*)\leqslant \lfloor \frac{n-k}{2}\rfloor$, which requires a complexity of $\mathcal{O}(\frac{5}{2}n^2-\frac{3}{2}k^2)$ in $\mathbb{F}_{q^m}$ \cite{loidreau2005welch}. In the last step, we shall compute $\phi_{\mathcal{B}}(\bm{e}^*)$ through the method described in Remark \ref{remark2} with a complexity of $\mathcal{O}(((m-1)(q-1)+1)mn)$ in $\mathbb{F}_{q^m}$. Finally the total complexity of decoding expanded Gabidulin codes is $\mathcal{O}(m^2n(q-1)+mn(3-q)+\frac{5}{2}n^2-\frac{3}{2}k^2)$ in $\mathbb{F}_{q^m}$ plus $\mathcal{O}(m^2n(n-k))$ in $\mathbb{F}_q$.

\section{Two hard problems}\label{section4}
The security of our two proposals in this paper mainly involves two hard problems in coding theory, namely the rank syndrom decoding (RSD) problem and MinRank problem. In this section, we will give a description of these two problems and some well known attacks on them.
\subsection{RSD problem}\label{section4.1}
\begin{definition}[RSD problem]
Let $H\in\mathcal{M}_{n-k,n}(\mathbb{F}_{q^m})$ be a matrix of full rank, $\bm{s}\in\mathbb{F}_{q^m}^{n-k}$ and $t$ be a positive integer. An RSD instance $\mathcal{R}(q,m,n,k,t)$ is to solve $\bm{s}=\bm{e}H^T$ for $\bm{e}\in\mathbb{F}_{q^m}^n$ such that $w_R(\bm{e})\leqslant t$.
\end{definition}

The RSD problem plays a crucial role in rank metric based cryptography. Although this problem is not known to be NP-complete, it is believed to be hard by the community. Up to now, the best known combinatorial attacks on this problem can be found in \cite{chabaud1996crypt,technique2002,complexity2016,algorithm2018}. In what follows, we will recall the principle of the combinatorial attack proposed in \cite{complexity2016}. Although there are some improvements \cite{algorithm2018} for this attack, they are not applicable to our proposals.

To make the description concise, here we introduce a notation used in the sequel. For positive integers $w<v<u$, by $\tl{P}_{\mathbb{F}_q}(u,v,w)$ we denote the probability that a random space of dimension $v$ in a space of dimension $u$ contains a given space of dimension $w$. Using the Gaussian binomial, we have
\[\tl{P}_{\mathbb{F}_q}(u,v,w)=\gaubinom{u-w}{v-w}\bigg/\gaubinom{u}{v}=\prod_{i=0}^{v-w-1}\frac{q^{u-w}-q^i}{q^{v-w}-q^i}\bigg/\prod_{i=0}^{v-1}\frac{q^u-q^i}{q^{v}-q^i}\sim \frac{1}{q^{w(u-v)}}.\]
For an RSD instance $\mathcal{R}(q,m,n,k,t)$, we consider the following two cases to solve the problem. 
\begin{itemize}
\item[]\textbf{Case $\bm{1}$:} $n> m$. Let $\mathcal{B}$ be a basis of $\mathbb{F}_{q^m}$ over $\mathbb{F}_q$, then there exists $E\in\mathcal{M}_{m,n}(\mathbb{F}_q)$ with $\rank(E)\leqslant t$ such that $\bm{e}=\mathcal{B}E$. Let $\mathcal{E}=\langle E^T\rangle_{\mathbb{F}_q}$ and $\mathcal{V}\subseteq\mathbb{F}_q^m$ be an $\mathbb{F}_q$-linear space of dimension $t'\geqslant t$. If $\mathcal{E}\subseteq\mathcal{V}$, then one can express $\bm{e}$ in a basis of $\mathcal{V}$ over $\mathbb{F}_q$. By computing $\bm{s}=\bm{e}H^T$ and expanding this system over the base field, one obtains a linear system of $m(n-k)$ equations and $nt'$ variables over $\mathbb{F}_q$. To have only one solution with overwhelming probability, one needs $nt'\leqslant m(n-k)$, then $t'\leqslant m-\left\lceil \frac{km}{n}\right\rceil$. By taking $t'=m-\left\lceil \frac{km}{n}\right\rceil$, one gets a complexity of $\mathcal{O}(m^3(n-k)^3/p)$ in $\mathbb{F}_q$, where $p=\tl{P}_{\mathbb{F}_q}(m,t',t)$.

\item[]\textbf{Case $\bm{2}$:} $n\leqslant m$. Let $\mathcal{B}$ be a basis of $\mathbb{F}_{q^m}$ over $\mathbb{F}_q$, then there exists $E\in\mathcal{M}_{m,n}(\mathbb{F}_q)$ with $\rank(E)\leqslant t$ such that $\bm{e}=\mathcal{B}E$. Let $\mathcal{E}=\langle E\rangle_{\mathbb{F}_q}$ and $\mathcal{V}\subseteq\mathbb{F}_q^n$ be an $\mathbb{F}_q$-linear space of dimension $t'\geqslant t$. If $\mathcal{E}\subseteq\mathcal{V}$, then one can express $E$ in a basis of $\mathcal{V}$ over $\mathbb{F}_q$. By computing $\bm{s}=\bm{e}H^T$ and expanding this system over the base field, one obtains a linear system of $m(n-k)$ equations and $mt'$ variables over $\mathbb{F}_q$. To have only one solution with overwhelming probability, one needs $mt'\leqslant m(n-k)$, then $t'\leqslant n-k$. By taking $t'=n-k$, one gets a complexity of $\mathcal{O}\left(m^3(n-k)^3/p\right)$ in $\mathbb{F}_q$, where $p=\tl{P}_{\mathbb{F}_q}(n,t',t)$.
\end{itemize}

\subsection{MinRank problem}
\begin{definition}[MinRank problem]
For a finite field $\mathbb{F}_q$ and positive integers $m,n,k,t$, let $M_1,\ldots,$\\$M_k\in\mathcal{M}_{m,n}(\mathbb{F}_q)$. A MinRank instance of parameters $(q,m,n,k,t)$ is to search for $x_1,\ldots,x_k\in\mathbb{F}_q$ such that $\tl{Rank}(\sum_{i=1}^kx_iM_i)\leqslant t$.
\end{definition}

The MinRank problem was first introduced and proven NP-complete by Buss et al. in \cite{computational1999}. This problem is of great importance in both multivariate cryptography \cite{cabarcas2017key} and rank metric based cryptography \cite{complexity2016}. In Table \ref{table1}, we give the best algebraic attacks in \cite{improvements2020} on the MinRank problem, where $\omega=2.8$ is the linear algebra constant.

\begin{table}[!h]\label{table1}
\setlength{\abovecaptionskip}{-0.1cm}
\setlength{\belowcaptionskip}{-0.2cm}
\begin{center}
\begin{tabular}{|c|l|}
\hline
 Condition & \makecell*[c]{Complexity}\\
\hline
\makecell*[c]{$A=k\binom{n}{t},B=m\binom{n}{t+1}$,\\ $A-1\leqslant B$}&$\mathcal{O}(BA^{\omega-1})$\\
\hline
\makecell*[c]{$A_b=\binom{n}{t}\binom{k+b-1}{b},B_b=\sum_{i=1}^b(-1)^{i+1}\binom{n}{t+i}\binom{m+i-1}{i}\binom{k+b-i-1}{b-i} $,\\$b=\min\{b:A_b-1\leqslant B_b\}$ and $q>b,b<t+2$}  &$\mathcal{O}\left(k(t+1)A_b^2\right)$\\
\hline
\makecell*[c]{$A_b=\sum_{j=1}^b\binom{n}{t}\binom{k}{j},B_b=\sum_{j=1}^b\sum_{i=1}^j(-1)^{i+1}\binom{n}{t+i}\binom{m+i-1}{i}\binom{k}{j-i}$,\\ $b=\min\{b:A_b-1\leqslant B_b\}$ and $q=2,b<t+2$} & $\mathcal{O}\left(k(t+1)A_b^2\right)$\\
\hline

\end{tabular}
\end{center}
\caption{Best algebraic attacks on the MinRank problem.}\label{table1}
\end{table}

\section{Description of our proposals}\label{section5}
Now we give a formal description of our two proposals.

\subsection{Proposal \Rmnum{1}}
For a given security level, choose a finite field $\mathbb{F}_q$ and positive integers $k<n\leqslant m$ and $\lambda$ such that $\frac{m(n-k)}{n}<\lambda<m$. Let $K=\lambda n-m(n-k)$ and $N=\lambda n$. For $0\leqslant j\leqslant n-1$, we define $I_j=\{mj+1,\ldots,mj+\lambda\}$ and let $S=\cup_{j=0}^{n-1}I_j$. Now we give a formal description of our first proposal through the following three procedures.

\begin{itemize}
\item Key generation\\
Let $\mathcal{G}\subseteq\mathbb{F}_{q^m}^n$ be an $[n,k]$ Gabidulin code. Randomly choose a basis $\mathcal{B}=(\alpha_1,\ldots,\alpha_m)$ of $\mathbb{F}_{q^m}$ over $\mathbb{F}_q$ and let $\widehat{\mathcal{G}}=\phi_\mathcal{B}(\mathcal{G})$. Denote by $\widehat{H}$ a parity-check matrix of $\widehat{\mathcal{G}}$ of the form $(\ref{paritycheck})$, and $\widehat{H}_S$ a submatrix of $\widehat{H}$ from the columns indexed by $S$. Let $\widehat{\mathcal{G}}_S=\langle\widehat{H}_S\rangle_{\mathbb{F}_q}^\perp$ and $\widehat{G}_S$ be a generator matrix of $\widehat{\mathcal{G}}_S$. Randomly choose $A\in GL_\lambda(\mathbb{F}_q)$ and set $T=I_n\otimes A$, where $I_n$ is the identity matrix of order $n$. Randomly choose $M\in GL_K(\mathbb{F}_q)$ such that $G_{pub}=M\widehat{G}_ST^{-1}$ is of systematic form. If such an $M$ does not exist, then repeat the process above. The public key is $(G_{pub},t)$ where $t=\lfloor\frac{n-k}{2}\rfloor$, and the secret key is $(\widehat{H}_S,A,\mathscr{D}_{\widehat{\mathcal{G}}})$ where $\mathscr{D}_{\widehat{\mathcal{G}}}$ is the fast decoder of $\widehat{\mathcal{G}}$.
\item Encryption\\
For a plaintext $\bm{x}\in\mathbb{F}_{q}^K$, randomly choose $E\in\mathcal{M}_{n,\lambda}(\mathbb{F}_q)$ with $\rank(E)=t$. Let $\bm{e}=(\bm{e}_1,\ldots,\bm{e}_n)\in\mathbb{F}_q^N$, where $\bm{e}_i$ is the $i$-th row vector of $E$. The ciphertext corresponding to $\bm{x}$ is computed as $\bm{y}=\bm{x}G_{pub}+\bm{e}$.
\item Decryption\\
For a ciphertext $\bm{y}\in\mathbb{F}_q^N$, let $\bm{e}'=\bm{e}T$ and compute
\begin{align*}
\bm{s}=\bm{y}T\widehat{H}_S^T=\bm{x}M\widehat{G}_ST^{-1}T\widehat{H}_S^T+\bm{e}T\widehat{H}_S^T=\bm{e}'\widehat{H}_S^T.
\end{align*} 
Applying $\mathscr{D}_{\widehat{\mathcal{G}}}$ to $\bm{s}$ will lead to a vector $\bm{e}''\in\mathbb{F}_q^{mn}$. The restriction of $\bm{e}''$ to $S$ will be $\bm{e}'$, then we can recover $\bm{e}$ by computing $\bm{e}'T^{-1}$. The plaintext will be the restriction of $\bm{y}-\bm{e}$ to the first $K$ coordinates.
\end{itemize}
\textbf{Correctness of Decryption.}
 Let $\bm{e}'=(\bm{e}'_1,\ldots,\bm{e}'_n)$, where $\bm{e}'_i=\bm{e}_iA$. Define $E'\in\mathcal{M}_{n,\lambda}(\mathbb{F}_q)$ to be a matrix whose $i$-th row vector is $\bm{e}'_i$. Let $\bm{e}''=(\bm{e}''_1,\ldots,\bm{e}''_n)$, where $\bm{e}''_i=(\bm{e}'_i,\bm{0})$ and $\bm{0}$ denotes the zero vector of length $m-\lambda$. Define $E''\in\mathcal{M}_{n, m}(\mathbb{F}_q)$ to be a matrix whose $i$-th row vector is $\bm{e}''_i$. It is easy to see that 
\begin{align*}
E''=[E'|0_{n\times (m-\lambda)}]=[EA|0_{n\times (m-\lambda)}].
\end{align*}
Then
\begin{align*}
\rank(E'')=\rank(E')=\rank(E)=t,
\end{align*}
which implies that $\bm{e}''$ satisfies the decodable condition described in Section \ref{condition}. Applying the fast decoder of $\widehat{\mathcal{G}}$ to $\bm{s}=\bm{e}'\widehat{H}_S^T=\bm{e}''\widehat{H}^T$ will lead to $\bm{e}''$, then the restriction of $\bm{e}''$ to $S$ will be $\bm{e}'$.

\subsection{Proposal \Rmnum{2}}
For a given security level, choose a finite field $\mathbb{F}_q$ and positive integers $\lambda\ll k<n\leqslant m$. Let $K=km$ and $N=nm$. Now we give a formal description of our second proposal through the following three procedures.

\begin{itemize}
\item Key generation\\
Let $\mathcal{G}\subseteq\mathbb{F}_{q^m}^n$ be an $[n,k]$ Gabidulin code. Randomly choose a basis $\mathcal{B}=(\alpha_1,\ldots,\alpha_m)$ of $\mathbb{F}_{q^m}$ over $\mathbb{F}_q$, and let $\widehat{\mathcal{G}}=\phi_\mathcal{B}(\mathcal{G})$ be an $[N,K]$ expanded code of $\mathcal{G}$. Let $\widehat{G}$ be a generator matrix of $\widehat{\mathcal{G}}$, and $\widehat{H}$ a parity-check matrix of the form $(\ref{paritycheck})$. Let $u_f=\lfloor\frac{n}{\lambda}\rfloor,u_c=\lceil\frac{n}{\lambda}\rceil$ and $v=n-\lambda u_f$. Randomly choose $A\in GL_{m\lambda}(\mathbb{F}_q)$ such that the $mv\times mv$ submatrix $A_{sub}$ in the top left corner of $A$ is invertible. Let
\[T=
\begin{pmatrix}
A_{ten}&\\
&A_{sub}
\end{pmatrix}\in GL_N(\mathbb{F}_q),
\] 
where $A_{ten}$ is the tensor product $I_{u_f}\otimes A$. If the first $K$ columns of $\widehat{G}T^{-1}$ are linearly independent over $\mathbb{F}_q$, then choose $M\in GL_K(\mathbb{F}_q)$ to convert $G_{pub}=M\widehat{G}T^{-1}$ into systematic form. Otherwise, one rechooses the matrix $T$. Then the public key is $(G_{pub},t)$ where $t=\lfloor\frac{n-k}{2\lambda}\rfloor$, and the private key is $(\widehat{H},A,\mathscr{D}_{\widehat{\mathcal{G}}})$ where $\mathscr{D}_{\widehat{\mathcal{G}}}$ is the fast decoder of $\widehat{\mathcal{G}}$.
\item Encryption\\
For a plaintext $\bm{x}\in\mathbb{F}_{q}^K$, randomly choose $E\in\mathcal{M}_{u_c, m\lambda}(\mathbb{F}_q)$ of the form
\begin{align}\label{matrixE}
E=
\begin{pmatrix}
\bm{e}_1&\cdots&\cdots&\cdots&\bm{e}_{\lambda-1}&\bm{e}_\lambda\\
\bm{e}_{1+\lambda}&\cdots&\cdots&\cdots&\bm{e}_{2\lambda-1}&\bm{e}_{2\lambda}\\
\vdots&\ddots&\ddots&\ddots&\vdots&\vdots\\
\bm{e}_{1+(u_c-1)\lambda}&\cdots&\bm{e}_n&\cdots&\bm{0}&\bm{0}\\
\end{pmatrix}
\end{align}
such that $\rank(E)=t$, where $\bm{e}_i\in\mathbb{F}_q^m$ for $1\leqslant i\leqslant n$. Let $\bm{e}=(\bm{e}_1,\ldots,\bm{e}_n)\in\mathbb{F}_q^N$, then the ciphertext corresponding to $\bm{x}$ is computed as $\bm{y}=\bm{x}G_{pub}+\bm{e}$.

\item Decryption\\
For a ciphertext $\bm{y}\in\mathbb{F}_q^N$, compute $\bm{s}=\bm{y}T\widehat{H}^T=\bm{e}T\widehat{H}^T$. Applying $\mathscr{D}_{\widehat{\mathcal{G}}}$ to $\bm{s}$ will lead to $\bm{e}'=\bm{e}T$, then one can recover $\bm{e}$ by computing $\bm{e}'T^{-1}$. The restriction of $\bm{y}-\bm{e}$ to the first $K$ coordinates will be the plaintext.
\end{itemize}
\textbf{Correctness of Decryption.} First, we introduce the following proposition.
\begin{proposition}\label{rank}
Given $\lambda$ matrices $M_1,M_2,\ldots,M_\lambda\in\mathcal{M}_{u, v}(\mathbb{F}_q)$, let 
\[
M=[M_1,M_2,\ldots,M_\lambda]
\ \tl{and}\ 
M'=
\begin{pmatrix}
M_1\\
M_2\\
\vdots\\
M_\lambda
\end{pmatrix}.
\]
Suppose $\rank(M)=t$, then there must be $\rank(M')\leqslant \lambda t$.
\end{proposition}

\begin{proof}
Note that $\rank(M)=t$, then $\rank(M_i)\leqslant t$ for $1\leqslant i\leqslant \lambda$. Hence $\rank(M')\leqslant\sum_{i=1}^\lambda\rank(M_i)\leqslant \lambda t$.
\end{proof}

From the decrypting process of Proposal \Rmnum{2}, it suffices to prove that $\bm{e}'$ satisfies the decodable condition described in Section \ref{condition}. Let $\bm{e}'=(\bm{e}'_1,\bm{e}'_2,\ldots,\bm{e}'_n)$, where $\bm{e}'_i\in\mathbb{F}_q^m$ for $1\leqslant i\leqslant n$, then there exist $\bm{e}'_{n+1},\ldots,\bm{e}'_{u_c\lambda}\in\mathbb{F}_q^m$ such that
\[E'=
\begin{pmatrix}
\bm{e}'_1&\cdots&\cdots&\cdots&\bm{e}'_{\lambda-1}&\bm{e}'_\lambda\\
\bm{e}'_{1+\lambda}&\cdots&\cdots&\cdots&\bm{e}'_{2\lambda-1}&\bm{e}'_{2\lambda}\\
\vdots&\ddots&\ddots&\ddots&\vdots&\vdots\\
\bm{e}'_{1+(u_c-1)\lambda}&\cdots&\bm{e}'_n&\bm{e}'_{n+1}&\cdots&\bm{e}'_{u_c\lambda}\\
\end{pmatrix}=EA=
[E'_1,E'_2,\ldots,E'_\lambda],
\]
where $E'_i\in\mathcal{M}_{u_c, m}(\mathbb{F}_q)$ for $1\leqslant i\leqslant \lambda$. Apparently we have $\rank(E')=\rank(E)=t$.
Let 
\[F=
\begin{pmatrix}
\bm{e}'_1\\
\bm{e}'_2\\
\vdots\\
\bm{e}'_n
\end{pmatrix}\ \tl{and}\ 
F'=
\begin{pmatrix}
E'_1\\
E'_2\\
\vdots\\
E'_\lambda
\end{pmatrix},
\]
then $\rank(F)\leqslant\rank(F')\leqslant \lambda t\leqslant \lfloor\frac{n-k}{2}\rfloor$ by Proposition \ref{rank}.

\begin{remark}
The cryptosystem presented above deals with the general situation where $\lambda$ does not divide $n$, or equivalently $u_f\neq u_c$. As for the case of $u_f=u_c$, just a few changes are needed in the key generation procedure. To generate the column scrambling matrix $T^{-1}$, any $A\in GL_{m\lambda}(\mathbb{F}_q)$ is feasible for computing $T=I_{u_f}\otimes A$.
\end{remark}

\section{Security analysis}\label{section6}
This section mainly discusses the security of the proposed cryptosystems. Attacks on code-based cryptosystems can be divided into two categories, namely structural attacks and generic attacks. A structural attack aims to recover the structure of the underlying code from the published information, which amounts to recovering the private key or its equivalent form that can be used to decrypt any valid ciphertext in polynomial time. A generic attack is to recover the plaintext directly without knowing the private key, which implies that one has to deal with the underlying hard problem.

In the remainder of this section, we will first introduce a distinguisher for Gabidulin codes and expalin why our proposals can prevent the existing structural attacks. Following this, we introduce the concept of twisted Frobenius power and build a distinguisher for expanded Gabidulin code, which provides an approach for us to distinguish expanded Gabidulin codes from random linear codes. After that, we will investigate the practical security of our proposals from two aspects.

\subsection{Existing structural attacks}\label{section5.1}
Now we describe some properties of Gabidulin codes under the Frobenius map, which will be useful for us to explain why our proposals can prevent the related structural attacks. The following two propositions provide an approach for us to distinguish Gabidulin codes from general ones.

\begin{proposition}\cite{gaborit2018polynomial}\label{proposition4}
Let $\mathcal{G}\subseteq\mathbb{F}_{q^m}^n$ be an $[n,k]$ Gabidulin code. For any positive integer $i$, the following equality holds
\begin{align*}
\dim(\mathcal{G}+\mathcal{G}^{[1]}+\cdots+\mathcal{G}^{[i]})=\min\{n,k+i\}.
\end{align*}
\end{proposition}

\begin{proposition}\cite{gaborit2018polynomial}\label{proposition5}
Let $\mathcal{C}\subseteq\mathbb{F}_{q^m}^n$ be an $[n,k]$ random linear code. For any positive integer $i$, the following equality holds with high probability
\begin{align*}
\dim(\mathcal{C}+\mathcal{C}^{[1]}+\cdots+\mathcal{C}^{[i]})=\min\{n,k(i+1)\}.
\end{align*}
\end{proposition}

Most cryptosystems based on Gabidulin codes have been shown to be insecure due to their vulnerability against structural attacks, such as Overbeck's attack\cite{2008Structural}, Coggia-Couvreu attack\cite{coggia2020security} and the attack proposed in \cite{gaborit2018polynomial}. Although these attacks were designed to cryptanalyze different variants, most of them rely on the fact that one can distinguish Gabidulin codes from general ones by observing how their dimensions behave under the Frobenius map according to Proposition \ref{proposition4} and \ref{proposition5}. However, this property is no longer valid when considering our proposals. Since our proposals are built over the base field $\mathbb{F}_q$, it is clear that $\widehat{\mathcal{G}}^{[i]}=\widehat{\mathcal{G}}$ for any integer $i$. In this situation, Gabidulin codes will be indistinguishable from random codes. Hence it is reasonable to conclude that all these attacks do not work on our proposals.

\subsection{A distinguisher for expanded Gabidulin codes}
Let $\mathcal{G}\subseteq\mathbb{F}_{q^m}^n$ be an $[n,k]$ Gabidulin code, and $\widehat{\mathcal{G}}$ an expanded code of $\mathcal{G}$ with respect to a basis $\mathcal{B}$ of $\mathbb{F}_{q^m}$ over $\mathbb{F}_q$. Given a generator matrix of $\widehat{\mathcal{G}}$, Berger et al. \cite{berger2017gabidulin} proposed an efficient approach to compute $\mathcal{G}^{[s]}$ for some $0\leqslant s\leqslant m-1$, which is also an $[n,k]$ Gabidulin code. A key point is that one can recover a normal generator matrix of $\widehat{\mathcal{G}}$, which can be done by reducing any generator matrix of $\widehat{\mathcal{G}}$ into systematic form. Hence the direct application of expanded Gabidulin codes will lead to an insecure scheme. Indeed, one can obtain more information from a generator matrix of $\widehat{\mathcal{G}}$. For instance, one can recover the expanded code of $\mathcal{G}^{[s]}$ without knowing the basis $\mathcal{B}$. To do this, we need to introduce the concept of twisted Frobenius power.

Let $G\in\mathcal{M}_{K,N}(\mathbb{F}_q)$, where $K=km$ and $N=nm$. For $1\leqslant i\leqslant k$ and $1\leqslant j\leqslant n$, let $I_i=\{(i-1)m+1,\ldots,im\}$ and $J_j=\{(j-1)m+1,\ldots,jm\}$. Denote by $G_{ij}$ the submatrix of $G$ from the rows indexed by $I_i$ and columns indexed by $J_j$. For a positive integer $s$, we define $G^{(s)}=(G_{ij}^{q^s})$ to be the $s$-th twisted Frobenius power of $G$, where $G_{ij}^{q^s}$ denotes the usual $q^s$-th power of $G_{ij}$. For a sequence $\bm{s}=(1,2,\ldots,N)$, we partition $\bm{s}$ into $n$ blocks, each of which has length $m$. Let $\mathcal{C}$ be an $[n,k]$ linear code over $\mathbb{F}_{q^m}$, and $\widehat{\mathcal{C}}$ an $[N,K]$ expanded code of $\mathcal{C}$. Let $I_1,\ldots,I_k$ be $k$ blocks of $\bm{s}$, if $I=\cup_{i=1}^kI_i$ forms an information set of $\widehat{\mathcal{C}}$, then we call $I$ a block information set of $\widehat{\mathcal{C}}$. It is clear that $\widehat{\mathcal{C}}$ admits at least one block information set. In the sequel, when we talk about a block information set of a linear code, we always mean the first one in lexicographic order. Let $I$ be a block information set of $\widehat{\mathcal{C}}$, and $G_I\in\mathcal{M}_{K,N}(\mathbb{F}_q)$ a generator matrix of $\widehat{\mathcal{C}}$, where the submatrix of $G_I$ from the columns indexed by $I$ forms an identity matrix of order $K$. Then we define the $s$-th twisted Frobenius power of $\widehat{\mathcal{C}}$ as $\widehat{\mathcal{C}}^{(s)}=\langle G_I^{(s)}\rangle_{\mathbb{F}_q}$. It is easy to see that $\widehat{\mathcal{C}}^{(s)}$ is an expanded code of $\mathcal{C}^{[s]}$ with respect to $\mathcal{B}$ and does not rely on the choice of the block information set.

The concept of twisted Frobenius power actually provides an approach for us to compute the expanded code of $\mathcal{C}^{[s]}$ even if the corresponding basis is not known. Furthermore, we have the following two propositions, which describe an effective distinguisher for expanded Gabidulin codes from an expanded code of a random one.

\begin{proposition}\label{egabidulin}\label{proposition7}
Let $\mathcal{G}\subseteq\mathbb{F}_{q^m}^n$ be an $[n,k]$ Gabidulin code, and $\widehat{\mathcal{G}}\subseteq\mathbb{F}_q^N$ an expanded code of $\mathcal{G}$ with respect to a basis $\mathcal{B}$ of $\mathbb{F}_{q^m}$ over $\mathbb{F}_q$. For any positive integer $i$, the following equality holds
\begin{align*}
\dim(\widehat{\mathcal{G}}+\widehat{\mathcal{G}}^{(1)}+\cdots+\widehat{\mathcal{G}}^{(i)})=\min\{N,(k+i)m\}.
\end{align*}
\end{proposition}
\begin{proof}
Note that $\widehat{\mathcal{G}}^{(i)}=\phi_{\mathcal{B}}(\mathcal{G}^{[i]})$, then
\begin{align*}
\widehat{\mathcal{G}}+\widehat{\mathcal{G}}^{(1)}+\cdots+\widehat{\mathcal{G}}^{(i)}&=\phi_{\mathcal{B}}(\mathcal{G})+\phi_{\mathcal{B}}(\mathcal{G}^{[1]})+\cdots+\phi_{\mathcal{B}}(\mathcal{G}^{[i]})\\
&=\phi_{\mathcal{B}}(\mathcal{G}+\mathcal{G}^{[1]}+\cdots+\mathcal{G}^{[i]}).
\end{align*}
It follows that
\[\dim(\widehat{\mathcal{G}}+\widehat{\mathcal{G}}^{(1)}+\cdots+\widehat{\mathcal{G}}^{(i)})=m\dim(\mathcal{G}+\mathcal{G}^{[1]}+\cdots+\mathcal{G}^{[i]}),\]
which leads to the conclusion immediately from Propostion \ref{proposition4}.
\end{proof}

\begin{proposition}\label{proposition8}
Let $\widehat{\mathcal{C}}\subseteq\mathbb{F}_q^N$ be an expanded code of an $[n,k]$ random code $\mathcal{C}\subseteq\mathbb{F}_{q^m}^n$. For any positive integer $i$, the following equality holds with high probability
\begin{align*}
\dim(\widehat{\mathcal{C}}+\widehat{\mathcal{C}}^{(1)}+\cdots+\widehat{\mathcal{C}}^{(i)})=\min\{N,k(i+1)m\}.
\end{align*}
\end{proposition}
\begin{proof}
Similar to the proof of Proposition \ref{egabidulin}, the conclusion is obtained immediately from Proposition \ref{proposition5}.
\end{proof}

An expanded code always has a block information set. However, a random code may not have one, such as a linear code $\mathcal{C}$ generated by a matrix $G$ that has a full-zero column in each block. If a linear code has a block information set, then we define its twisted Frobenius power following the way described above. Otherwise, we choose a generator matrix of $\mathcal{C}$ at random, say $G$, and compute $\mathcal{C}^{(s)}=\langle G^{(s)}\rangle_{\mathbb{F}_q}$. In both situations, however, $\mathcal{C}^{(s)}$ generally depends on the choice of the block information set or the generator matrix $G$. Furthermore, we have a heuristic from Proposition \ref{proposition8} described as follows, which has been verified practically through numerous experiments.\\

\noindent \textbf{Heuristic}. 
\textit{Let $\mathcal{C}\subseteq\mathbb{F}_q^N$ be an $[N,K]$ random code, where $N=nm$ and $K=km$. For any positive integer $i$, the following equality holds with high probability
\begin{align*}
\dim(\widehat{\mathcal{C}}+\widehat{\mathcal{C}}^{(1)}+\cdots+\widehat{\mathcal{C}}^{(i)})=\min\{N,k(i+1)m\}.
\end{align*}}

With this heuristic and Propositions \ref{proposition7},\ref{proposition8}, we actually build a distinguisher for expanded Gabidulin codes. By definition, the applicable condition for computing the twisted Frobenius power of a code is that both the length and dimension should be multiples of the extension degree. In our implementation, we choose $m=n$ and therefore this applicable condition is satisfied. According to our experimental results, both the public code and its dual in our two proposals behave more like a random code. Therefore, we believe that these two proposals can prevent a potential attack based on this distinguisher.

\subsection{Generic attacks}\label{section5.2}
In this section, we will show that decrypting a valid ciphertext in our proposals can be converted into solving a MinRank instance. To facilitate the description of problems, we introduce an $\mathbb{F}_q$-linear isomorphism $\sigma_n$ from $\mathbb{F}_q^{ns}$ to $\mathcal{M}_{n,s}(\mathbb{F}_q)$. For a vector $\bm{x}=(\bm{x}_1,\ldots,\bm{x}_n)\in\mathbb{F}_q^{ns}$ with $\bm{x}_i\in\mathbb{F}_q^s$, we define $\sigma_n(\bm{x})$ as
\[\sigma_n(\bm{x})=\begin{bmatrix}\bm{x}_1^T,\ldots,\bm{x}_n^T\end{bmatrix}^T\in\mathcal{M}_{n,s}(\mathbb{F}_q).\]
For a set $\mathcal{X}\subseteq\mathbb{F}_q^{ns}$, we define $\sigma_n(\mathcal{X})=\{\sigma_n(\bm{x}):\bm{x}\in\mathcal{X}\}$. For any $\bm{x}\in\mathbb{F}_q^{ns}$, by $w_R(\bm{x})$ we mean the usual rank of $\sigma_n(\bm{x})$ hereafter when no ambiguity arises.\\

\textbf{Reducing Proposal \Rmnum{1} to the MinRank problem.} In Proposal \Rmnum{1}, $K=\lambda n-m(n-k),N=\lambda n$ and $t=\lfloor\frac{n-k}{2}\rfloor$. Let $\bm{y}=\bm{c}+\bm{e}$ be a valid ciphertext in Proposal \Rmnum{1}, where $\bm{c}\in\mathcal{G}_{pub}$ and $\bm{e}\in\mathbb{F}_q^N$ with $w_R(\bm{e})\leqslant t$. Let $M_0=\sigma_n(\bm{y})$, and $M_i=\sigma_n(\bm{m}_i)$ where $\bm{m}_i$ denotes the $i$-th row of $G_{pub}$ for $1\leqslant i\leqslant K$. Then recovering $\bm{e}$ can be reduced to a MinRank instance  of parameters $(q,n,\lambda,K+1,t)$, that is to search for $a_0,a_1,\ldots,a_K\in\mathbb{F}_q$ such that $\rank(\sum_{i=0}^Ka_iM_i)\leqslant t$.\\

\textbf{Reducing Proposal \Rmnum{2} to the MinRank problem.} In Proposal \Rmnum{2}, $K=km,N=nm$ and $t=\lfloor\frac{n-k}{2\lambda}\rfloor$. Let $\bm{y}=\bm{c}+\bm{e}$ be a valid ciphertext in Proposal \Rmnum{2}, where $\bm{c}\in\mathcal{G}_{pub}$ and $\bm{e}\in\mathbb{F}_q^N$ with $w_R(\bm{e})\leqslant \lambda t$. Let $M_0=\sigma_n(\bm{y})$, and $M_i=\sigma_n(\bm{m}_i)$ where $\bm{m}_i$ denotes the $i$-th row of $G_{pub}$ for $1\leqslant i\leqslant K$. Then recovering $\bm{e}$ can be reduced to a MinRank instance of parameters $(q,n,m,K+1,\lambda t)$, that is to search for $a_0,a_1,\ldots,a_K\in\mathbb{F}_q$ such that $\rank(\sum_{i=0}^Ka_iM_i)\leqslant \lambda t$.

\subsubsection{Combinatorial attacks}
Based on the idea of combinatorial attacks on the RSD problem described in Section \ref{section4.1}, we now propose a combinatorial attack to evaluate the practical security our two proposals.\\

%\noindent
\textbf{Proposal \Rmnum{1}.} 
In the case of $n>\lambda$, let $\mathcal{E}=\langle E\rangle_{\mathbb{F}_q}$ and $\mathcal{V}\subseteq\mathbb{F}_q^\lambda$ be an $\mathbb{F}_q$-space of dimension $t'\geqslant t$. If $\mathcal{E}\subseteq\mathcal{V}$, then one can express $E$ in a basis of $\mathcal{V}$ over $\mathbb{F}_q$. With the same analysis as Case $1$ in Section \ref{section4.1}, one can obtain a linear system of $m(n-k)$ equations with $nt'$ variables in $\mathbb{F}_q$. Let $m(n-k)\geqslant nt'$, then $t'\leqslant m-\lceil\frac{km}{n}\rceil$. By taking $t'=m-\left\lceil\frac{km}{n}\right\rceil$, one gets a complexity of $\mathcal{O}\big(m^3(n-k)^3q^{t(\lambda-m+\lceil\frac{km}{n}\rceil)}\big)$. 

In the case of $n\leqslant\lambda$, let $\mathcal{E}=\langle E^T\rangle_{\mathbb{F}_q}$ and $\mathcal{V}\subseteq\mathbb{F}_q^n$ be an $\mathbb{F}_q$-space of dimension $t'\geqslant t$. If $\mathcal{E}\subseteq \mathcal{V}$, then one can express $E$ in a basis of $\mathcal{V}$ over $\mathbb{F}_q$. With the same analysis as Case $2$ in Section \ref{section4.1}, one can obtain a linear system of $m(n-k)$ equations with $\lambda t'$ variables in $\mathbb{F}_q$. Let $m(n-k)\geqslant \lambda t'$, then $t'\leqslant \lfloor\frac{m(n-k)}{\lambda}\rfloor$. By taking $t'=\lfloor\frac{m(n-k)}{\lambda}\rfloor$, one gets a complexity of $\mathcal{O}\big(m^3(n-k)^3q^{t(n-\lfloor\frac{m(n-k)}{\lambda}\rfloor)}\big)$.\\

%\noindent
\textbf{Proposal \Rmnum{2}.} 
On the one hand, $\mathcal{E}=\langle \bm{e}_1,\ldots,\bm{e}_n\rangle_{\mathbb{F}_q}$ has dimension at most $\lambda t$. Because of $n\leqslant m$, with the same analysis as Case $2$ in Section \ref{section4.1}, one gets a complexity of $\mathcal{O}\big(m^3(n-k)^3q^{\lambda tk}\big)$.

On the other hand, $\bm{e}$ is obtained from $E\in\mathcal{M}_{u_c, m\lambda}(\mathbb{F}_q)$ with $\rank(E)=t$. For $\lambda\geqslant 2$, it is clear that $u_c=\lceil\frac{n}{\lambda}\rceil< n< m\lambda$. Let $\mathcal{E}=\langle E^T\rangle_{\mathbb{F}_q}$ and $\mathcal{V}\subseteq\mathbb{F}_q^{u_c}$ be an $\mathbb{F}_q$-space of dimension $t'\geqslant t$. If $\mathcal{E}\subseteq\mathcal{V}$, then one can express $E$ in a basis of $\mathcal{V}$ over $\mathbb{F}_q$. With the same analysis as Case $2$ in Section \ref{section4.1}, one can obtain a linear system of $m(n-k)$ equations and $m\lambda t'$ variables in $\mathbb{F}_q$. Let $m(n-k)\geqslant m\lambda t'$, then $t'\leqslant \lfloor\frac{n-k}{\lambda}\rfloor$. By taking $t'=\lfloor\frac{n-k}{\lambda}\rfloor$, one gets a complexity of $\mathcal{O}(m^3(n-k)^3q^{t(u_c-\lfloor\frac{n-k}{\lambda}\rfloor)})$.

\subsubsection{Algebraic attacks}
Having reduced the proposals to a MinRank instance, one can directly apply the algebraic attacks described in Table \ref{table1} to evaluate the practical security.

\section{Parameters and performance}\label{section7}
In this section, we compute the public key size and information rate of the proposed cryptosystems for the security of 128 bits, 192 bits and 256 bits against the generic attacks described in Section \ref{section5.2}. After that we will make a comparison on public key size with some other code-based cryptosystems.

In Proposal \Rmnum{1}, the public key is a systematic generator matrix of a $[\lambda n,\lambda n-mr]$ code where $r=n-k$, resulting in a public key size of $mr(n\lambda-mr)\cdot log_2(q)$ bits. In Proposal \Rmnum{2}, the public key is a systematic generator matrix of an $[mn,mk]$ code, resulting in a public key size of $rkm^2\cdot log_2(q)$ bits. As for the information rate, this value is evaluated as $(n\lambda-mr)/n\lambda$ for Proposal \Rmnum{1}, and $k/n$ for Proposal \Rmnum{2} respectively.

\begin{table}[h!]
\setlength{\abovecaptionskip}{-0.2cm}
\setlength{\belowcaptionskip}{-0.2cm}
\begin{center}
\begin{tabular}{|c|ccccc|r|c|c|}
\hline      
\multirowcell{2}{\makecell*[c]{Instance}} & \multicolumn{5}{c|}{\makecell*[c]{Parameters}} & \multirowcell{2}{\makecell*[c]{Key Size}} & \multirowcell{2}{\makecell*[c]{Rate}} & \multirowcell{2}{\makecell*[c]{Security}}\\
\cline{2-6}
&\makecell*[c]{$q$}&$m$&$n$&$k$&$\lambda$&&&\\
\hline
\multirowcell{3}{Proposal \Rmnum{1}}
&2&31&31&19&29&24506&0.59&128\\
&2&38&38&20&36&58482&0.50&192\\
&2&45&45&25&43&116438&0.53&256\\
\hline
\multirowcell{3}{Proposal \Rmnum{2}}
&2&56&56&28&2&307328&0.49&128\\
&2&72&72&32&2&829440&0.44&192\\
&2&84&84&40&2&1552320&0.48&256\\
\hline
\end{tabular}
\end{center}
\caption{Public key size (in bytes) and information rate.}\label{table2}
\end{table}

\begin{table}[h!]
\setlength{\abovecaptionskip}{-0.2cm}
\setlength{\belowcaptionskip}{-0.2cm}
\begin{center}
\begin{tabular}{|c|ccccc|r|c|c|}
\hline      
\multirowcell{2}{\makecell*[c]{Instance}} & \multicolumn{5}{c|}{\makecell*[c]{Parameters}} & \multirowcell{2}{\makecell*[c]{Key Size}} & \multirowcell{2}{\makecell*[c]{Rate}} & \multirowcell{2}{\makecell*[c]{Security}}\\
\cline{2-6}
&\makecell*[c]{$q$}&$m$&$n$&$k$&$\lambda$&&&\\
\hline
\multirowcell{3}{Proposal \Rmnum{1}}
&7&20&20&12&18&11230&0.56&128\\
&7&24&24&14&22&24256&0.55&192\\
&7&28&28&16&26&46221&0.54&256\\
\hline
\multirowcell{3}{Proposal \Rmnum{2}}
&7&35&35&23&2&118646&0.66&128\\
&7&45&45&29&2&329724&0.64&192\\
&7&51&51&31&2&565900&0.61&256\\
\hline
\end{tabular}
\end{center}
\caption{Public key size (in bytes) and information rate.}\label{table3}
\end{table}

\begin{table}[h!]
\setlength{\abovecaptionskip}{-0.2cm}
\setlength{\belowcaptionskip}{-0.2cm}
\begin{center}
\begin{tabular}{|c|ccccc|r|c|c|}
\hline      
\multirowcell{2}{\makecell*[c]{Instance}} & \multicolumn{5}{c|}{\makecell*[c]{Parameters}} & \multirowcell{2}{\makecell*[c]{Key Size}} & \multirowcell{2}{\makecell*[c]{Rate}} & \multirowcell{2}{\makecell*[c]{Security}}\\
\cline{2-6}
&\makecell*[c]{$q$}&$m$&$n$&$k$&$\lambda$&&&\\
\hline
\multirowcell{3}{Proposal \Rmnum{1}}
&13&18&18&12&16&8993&0.63&128\\
&13&21&21&11&19&18359&0.47&192\\
&13&25&25&15&23&37583&0.57&256\\
\hline
\multirowcell{3}{Proposal \Rmnum{2}}
&13&29&29&17&2&79358&0.59&128\\
&13&37&37&21&2&212768&0.57&192\\
&13&43&43&23&2&393422&0.53&256\\
\hline
\end{tabular}
\end{center}
\caption{Public key size (in bytes) and information rate.}\label{table4}
\end{table}

We investigate the performance of the proposed cryptosystems in three cases, as summarized in Tables \ref{table2}, \ref{table3} and \ref{table4}. Taking the public key size as a more important issue, we suggest the parameter sets in Table \ref{table4} for the corresponding security level. With such a choice of parameter sets, we make a comparison on public key size with some other code-based cryptosystems in Table \ref{table5}.

\begin{table}[h!]
\setlength{\abovecaptionskip}{-0.2cm}
\setlength{\belowcaptionskip}{-0.2cm}
\begin{center}
\begin{tabular}{|c|r|r|r|}
\hline
\backslashbox{Instance}{Security} & \makecell*[c]{128} & \makecell*[c]{192} & \makecell*[c]{256}\\
\hline\rule{0pt}{10pt}
\makecell*[c]{Classic McEliece \cite{daniel2020classic}} & 261120 & 524160 & 1044992\\
\hline\rule{0pt}{10pt}
\makecell*[c]{BGR \cite{berger2017gabidulin}} & 11571 &  & 47104\\
\hline\rule{0pt}{10pt}
\makecell*[c]{HQC \cite{melchor2020hammi}} & 2249 & 4522 & 7245\\
\hline\rule{0pt}{10pt}
\makecell*[c]{BIKE \cite{aragon2020bikeb}} & 1540 & 3082 & 5121\\
\hline\rule{0pt}{10pt}
\makecell*[c]{Proposal \Rmnum{1}} & 8993 & 18359 & 37583\\
\hline\rule{0pt}{10pt}
\makecell*[c]{Proposal \Rmnum{2}} & 79358 & 212768 & 393422\\
\hline
\end{tabular}
\end{center}
\caption{Comparison on public key size (in bytes).}\label{table5} 
\end{table}

\section{Conclusion}\label{section8}
This paper presents two cryptosystems by using expanded Gabidulin codes in the McEliece setting. In our proposals, the underlying expanded Gabidulin code is divided into $n$ blocks, with each block having size $m$. By definition, each block corresponds to one component of the parent Gabidulin code. To destroy this correspondence, in Proposal \Rmnum{1} we first shorten the expanded Gabidulin code, then perform a column-mixing transformation to each block. In Proposal \Rmnum{2}, we adopt a rather different column-mixing transformation to the underlying code by mixing $\lambda$ neighbouring blocks. The greatest advantage of using expanded Gabidulin codes is that all existing structural attacks based on the Frobenius map no longer make sense. Additionally, an effective distinguisher is introduced for expanded Gabidulin codes, and the public code in these two proposals seems indistinguishable from random codes under this distinguisher. Furthermore, our proposals have a clear advantage in public key representation over some other code-based cryptosystems. For instance, we have reduced the public key size by around 96\% compared to Classic McEliece entering the third round of the NIST PQC project.

\begin{acknowledgements}
This research is supported by the National Key Research and Development Program of China (Grant No. 2018YFA0704703), the National Natural Science Foundation of China (Grant No. 61971243), the Natural Science Foundation of Tianjin (20JCZDJC00610), and the Fundamental Research Funds for the Central Universities of China (Nankai University).
\end{acknowledgements}

%\bibliographystyle{unsrt}
%\bibliography{reference}
\end{document}